\documentclass[12pt]{article}
\usepackage{amsmath}
\usepackage{times}
\usepackage{graphicx}
\usepackage{color}
\usepackage{multirow}
\usepackage[authoryear]{natbib}
\usepackage[colorlinks=true, allcolors=blue]{hyperref}
\usepackage{rotating}
\usepackage{bbm}
\usepackage{latexsym}
\usepackage{cleveref}
\usepackage{tikz-cd}
\tikzcdset{scale cd/.style={every label/.append style={scale=#1},
    cells={nodes={scale=#1}}}}
\usepackage{float}

\usepackage{mathtools}
\usepackage{amsthm, amssymb}
\usepackage{times}
\usepackage{graphicx}
\usepackage{subcaption}

\newtheorem{teo}{Theorem}[section]
\newtheorem{defi}[teo]{Definition}
\newtheorem{prop}[teo]{Proposition}

\newcommand{\R}{\mathbb{R}}
\newcommand{\C}{\mathbb{C}}

\newcommand{\captionfonts}{\normalsize}

\makeatletter  
\long\def\@makecaption#1#2{%
  \vskip\abovecaptionskip
  \sbox\@tempboxa{{\captionfonts #1: #2}}%
  \ifdim \wd\@tempboxa >\hsize
    {\captionfonts #1: #2\par}
  \else
    \hbox to\hsize{\hfil\box\@tempboxa\hfil}%
  \fi
  \vskip\belowcaptionskip}
\makeatother   
\begin{document}

\hspace{13.9cm}

\ \vspace{20mm}\\

{\LARGE A Categorical Framework for Quantifying\\Emergent Effects in Network Topology }

{\bf \large Johnny Jingze Li$^{\displaystyle 1,\displaystyle 2}$, Sebastian Pardo Guerra$^{\displaystyle 1, \displaystyle 2, \displaystyle 3}$, Kalyan Basu$^{\displaystyle 1}$, Gabriel A. Silva$^{\displaystyle 1, \displaystyle 3, \displaystyle 4}$}\\
{$^{\displaystyle 1}$Center for Engineered Natural Intelligence, University of California San Diego, La Jolla, CA, United States}\\
{$^{\displaystyle 2}$Department of Mathematics, University of California San Diego, La Jolla, CA, United States}\\
{$^{\displaystyle 3}$Department of Bioengineering, University of California San Diego, La Jolla, CA, United States}\\
%{$^{\displaystyle 4}$Microsoft, Technology \& Research, Redmond, WA, USA}\\
{$^{\displaystyle 4}$Department of Neurosciences, University of California San Diego, La Jolla, CA, United States}\\

%

%\ \\[-2mm]

%\tableofcontents

\thispagestyle{empty}
\markboth{}{NC instructions}
\ \vspace{-0mm}\\
\newpage
%Abstract
\begin{center} {\bf Abstract} \end{center}
Emergent effect is crucial to understanding the properties of complex systems that do not appear in their basic units, but there has been a lack of theories to measure and understand its mechanisms. In this paper, we consider emergence as a kind of structural nonlinearity, discuss a framework based on homological algebra that encodes emergence as the mathematical structure of cohomologies, and then apply it to network models to develop a computational measure of emergence. This framework ties the potential for emergent effects of a system to its network topology and local structures, paving the way to predict and understand the cause of emergent effects.  We show in our numerical experiment that our measure of emergence correlates with the existing information-theoretic measure of emergence.  
%%%%%%%%%%%

\section{Introduction}
Emergent effects of complex systems are, broadly speaking, phenomena present in the system that are not shared by and cannot be explained by an understanding or considerations of the system's constituent components in isolation. This is often expressed as 'the whole is more than the sum of its parts'. The concept of emergence in complex systems has been a topic of extensive study and is of importance across most of the natural and physical sciences, and increasingly in engineering \citep{anderson1972more,kalantari2020emergence, post1997emergent, o2020emergent,forrest1990emergent}. 

While there is no universally accepted technical definition of emergence, there are a number of intuitive ideas and concepts that capture its essence and find commonality across disciplines. Probably the most important is the notion of novel properties, patterns, or behaviors that are not explicitly or predictably present in the constituent parts that make up a system, or in the interactions between the parts. This phenomenon captures the idea that emergent properties are not reducible to or explainable solely by the properties of individual components, but rather, arise from the interactions and organization of those components within the system in highly non-obvious ways. It also emphasizes the non-linearity and unpredictability associated with emergence.  

Surprisingly, with the recent performance of large language models, emergence is also becoming an important topic in machine learning and artificial intelligence. Researchers have observed a surge in performance when the scale of these models reached a certain threshold that could not be predicted by extrapolation \citep[see, e.g.,][]{wei2022emergent,teehan2022emergent}. The observations of emergent abilities of machine learning model as it scales up is opening up novel and exciting research avenues to explore, and inspire the quest to explain, the unreasonable effectiveness of massive deep learning architectures. 

Still, a clear limitation of current work on emergence across essentially all fields and applications is that they remain at a phenomenological level, such as the shift of system behavior with an increasing number of computational units. The kind of measures derived from observation of such macroscopic effects do not result in mechanistic insights into why such emergence occurs, or how to anticipate and predict it. To address this, we need new measures that can describe and predict emergence on a more fundamental level capable of incorporating structural information and interactions between components \cite[see, e.g.,][]{barabasi2004network}. Grounded in a mechanistic understanding of emergence with measures that capture causal relationships, new theoretical models and numerical simulations might be able to explain how system structure changes when emergent effects occur. Eventually, this might support the intentional engineering of emergent effects. 

There have been a number of attempts at establishing frameworks that describe and measure emergence \citep[e.g.][]{crutchfield1994calculi, bar2004mathematical, gadioli2018categoriacal, gershenson2012complexity, correa2020metrics, green2023emergence, mediano2022greater}. Some approaches are based on information-theoretic notions, but in most cases the structure and interactions between participating components are not explicitly included in the computation of emergence, thus falling short of providing mechanistic insights. The challenge is in connecting the level of observations where researchers study the complexity of structures, patterns or dynamics with the more fundamental level of interacting components that give rise to emergent effects, and how to represent interactions in a way that is effective and easy to compute. As such, a universal and computational measure of emergence applicable to real-world systems still does not exist. 

Notably, in \cite{adam2017systems}, these shortcomings are partially addressed by studying the functor that connects two levels and proving that the approach to quantify emergence as the non-commutativity of algebraic operations is equivalent to a "loss of exactness" which we progressively describe in our paper — in the sense that the interactions between components are represented through the morphisms and colimits of a category. We show that observing the algebraic properties of the derived functor on that category is sufficient to capture and describe emergence. This new perspective thus supports more effective ways of quantifying emergence in a practical setting.

Our work in this paper builds on the work by \cite{adam2017systems}. Specifically, we establish a framework for studying and evaluating emergence — which we refer to as generativity to be more precise — that incorporates both an abstract mathematical framework and a practical computational metric capable of being applied to real-world systems such as networks. For example, we will show that our framework can potentially be used to study the emergent properties of massive deep learning architectures. Essentially, a system's ability to sustain emergent effects is represented as the mathematical structure of derived functors and cohomologies which are computable. Algebraic tools can then be applied to further investigate emergent effects. Compared with numerical metrics that evaluate emergence, our framework is richer in information, as the algebraic structure encodes more structural information about the system, for example, the contribution of each component to emergence as the global property of the system, thus supporting an increased understanding of the mechanistic details resulting in emergent phenomena. We then compare this framework with the existing information-theoretic measure of emergence through the numerical experiment on networks of different connectivities and show a correlation that supports the ability of our theory to capture emergence generally. 

The paper will be organized as follows: In Section 2 we define and illustrate emergent effects through examples. In Section 3, we present the framework for emergence as loss of exactness, adapted from \cite{adam2017systems}. In Section 4 and 5, this framework is realized as the mapping between networks, or more generally, quiver representations and then we compute the mathematical structure which encodes emergence. In Section 6 we establish the numerical measure of emergence and then discuss the numerical result showing a correlation with existing measures of emergence.

%\section{A Descriptive Framework for Emergent Effects}
\section{Emergence as Structural Nonlinearity}

Two key conceptual components are necessary to qualitatively describe emergent effects within the framework proposed by Adam \cite[see][]{adam2017systems}. The first is a notion of interaction or interconnection among the components of a system. For example, a thermodynamics system consisting of particles exhibits interesting properties because of the interactions among particles. The second is the notion of interactional effects, which equips each system with an observable. For example, we can measure global properties of particles systems, like temperature, but not from the observation or description at the level of individual particles. These kinds of interactional effects are almost always associated with partial observations, or a simplification and integration of lower — more foundational or granular —levels or scales in the system that result in a 'loss of information' or inability to properly observe the effects of those interactions at a higher level. 

With these two ingredients, we can define emergence as a partial observation of interacting and interconnected components within a system that cannot be explained by known interactions that produce or result in partial observations of the components. This notion agrees with the intuitive understanding of emergence that some properties of the interconnected components cannot be decomposed or reduced to combinations of known properties of the constituent components, i.e. that the whole is more than the sum of its parts. This notion of emergence is the foundation on which our work in this paper, building on the framework first proposed in \cite{adam2017systems}, develops a mathematical definition and computational measure of emergence. 

To formalize these ideas, we begin by representing the interactions between components as an operation $\vee$, where $s_1\vee s_2$ represents a new interconnected system of subsystems $s_1$ and $s_2$. Interactional effects are described by the mapping $\Phi$ that sends a system to its partial observation or interactional effect at a larger scale, in some cases corresponding to a coarse-graining scheme \cite{rosas2024software}. Emergent effects are sustained whenever the observation of the combined system cannot be explained by the observation of the separate components.  Mathematically,
\begin{defi}
A system sustains emergent effects when the following inequality is satisfied:
\begin{align}
\Phi(s_1\vee s_2) \neq \Phi(s_1)\vee \Phi(s_2),
\end{align} 
for some constituent subsystems $s_1$ and $s_2$.\label{emergencedef} 
\end{defi}

Let's consider the simple case when $\Phi$ is simply a smooth function $f: \R \rightarrow \R$ and the interaction $\vee$ is simply taking the average, $s_1\vee s_2 = (s_1+s_2)/2$. Then we realize that the extent to which $\Phi(s_1\vee s_2)$ differs from $\Phi(s_1)\vee \Phi(s_2)$ is just $\Big |f(\frac{s_1+s_2}{2})-\frac{f(s_1)+f(s_2)}{2} \Big|$, which is related to how nonlinear the function $f$ is, and can be studied by the derivatives of $f$, in particular, the second order derivative, since $\Big |f(\frac{s_1+s_2}{2})-\frac{f(s_1)+f(s_2)}{2} \Big|$ can be approximated by $\frac{|s_2-s_1|^2}{4}|f''(\xi)|$ for some $\xi$ between $s_1$ and $s_2$. Now when $\Phi$ models the cross-scale information flow in real world systems, we want an analogue to derivatives to apply this idea, and this naturally leads to the concept of a derived functor in homological algebra, which will be discussed in the next Section. We can also see that Definition \ref{emergencedef} captures the structural nonlinearity in emergence, the nonlinearity of system's behavior and functionality as the system's structure changes, or as we go from components to parts of the system to the whole system.

This is a general mathematical definition of emergence, as first given in \cite{adam2017systems} \footnote{In \cite{adam2017systems} the term "generativity" is used instead of emergence. While these two terms are often considered interchangeable, we opt for "emergence" in this paper, given its broader recognition and usage.}. Note that when studying the emergence of a specific system, the interaction $\vee$ and interactional effect $\Phi$ need to be chosen carefully. Some examples are provided in \cite{adam2017systems} and we will also discuss one in this section. Moreover, although this inequality is fundamental, additional constraints can be included in the definition in order to reflect the general notions of emergent effects in fields like condensed matter physics and system biology. 

The reader should note that while similar qualitative notions of emergence have been used in other work \citep{bar2004mathematical, gadioli2018categoriacal,   green2023emergence, mediano2022greater}, their use of these concepts to describe a formal notion of emergence falls short. As far as we know, no prior attempts have resulted in a quantity based on Definition \ref{emergencedef} that directly produces a measure of emergence. \cite{adam2017systems} attempted to address this issue, and we will discuss it in the next section by proposing emergence as "the loss of exactness", a term from homological algebra.

We will discuss the theoretical framework for measuring the difference between $\Phi(s_1\vee s_2)$ and $\Phi(s_1)\vee \Phi(s_2)$ in Section 2, and the final measure will be given in Section 6.  We hereby provide three examples, in machine learning, biology and statistical mechanics to illustrate this idea of emergence. 

\subsection{Machine Learning}

Emergence, or generativity, is an increasingly important concept in machine learning \citep[for example,][]{wei2022emergent, wei2022chain, du2024understanding}. Emergent abilities of large language models, for example, \cite{wei2022emergent}, commonly conceived as the new properties/ abilities of the larger models that do not exist in smaller models. If we consider $s_1$ and $s_2$ as two smaller models, $s_1\vee s_2$ as combining two smaller models into a larger model by, for example, techniques in ensemble learning, \citep[see, e.g.,][]{mohammed2023comprehensive}, and $\Phi$ as the mapping that reflects the properties/ abilities of the model, that is, $\Phi(s)$ is the ability acquired by the model $s$. Then  $\Phi(s_1\vee s_2)$ is the properties/ abilities of the combined model and  $\Phi(s_1)\vee \Phi(s_2)$ can be interpreted as a summation of the properties/ abilities of each small model. Then the difference between $\Phi(s_1\vee s_2)$ and $\Phi(s_1)\vee \Phi(s_2)$ can reflect the emergent properties/ abilities that result in the nonlinear increase of performance, related to the performance in \cite{wei2022emergent}. The difference can also be related to generalizability, where $s_1$ and $s_2$ are two data sets, when the model trains on two data sets, it is usually different from training the model on separate datasets.

\subsection{Biology}

In biological systems, the concept of emergence is exemplified through the self-organizing behaviors observed in social insects, such as ants or bees \cite{moffett2021ant}. The behaviors of individual ants are governed by a set of simple rules driven by genetic programming and environmental interactions. Pheromones in particular are key chemicals involved in communication and the exchange of information between ants in a colony. However, when these individuals interact as part of a larger group, the collective behavior of the colony transcends the capabilities of any single ant. The emergent properties of the colony, such as complex nest building, efficient foraging strategies, and intricate social organization, arise from the dynamic interplay and feedback loops between the ants. Each ant's contributions, though simple in isolation, collectively result in a higher-level organization that is qualitatively different from and cannot be predicted merely by understanding the behavior of individual ants. Consider $s_1$ and $s_2$ as two entities, for example, two groups of ants, and $s_1 \vee s_2$ represents the large group formed by $s_1$ and $s_2$. Let $\Phi(s)$ represent the colony structure formed by the ant group $s$. Then $\Phi(s_1\vee s_2)$ represents the colony structure built through the interaction among the ants, which exhibits complex structures compared to $\Phi(s_1)\vee \Phi(s_2)$, the colony structure built without the interaction between $s_1$ and $s_2$. The difference between $\Phi(s_1\vee s_2)$ and $\Phi(s_1)\vee \Phi(s_2)$ reflects the emergence exhibited in the example of ant colonies.

\subsection{Statistical Mechanics}

In the domain of statistical physics, the concept of emergent phenomena is profoundly illustrated by the study of phase transitions \cite{goldenfeld2018lectures}, where macroscopic changes emerge from the collective dynamics of microscopic constituents. Phase transitions, such as the crystallization of a liquid or the magnetization of a ferromagnetic material, represent a paradigmatic example of how large-scale, qualitative changes in a system's properties can arise from the cooperative interactions among its many individual components. Theoretical models, such as the Ising model for ferromagnetism, exemplify this: individual atomic spins interact locally, but when the system reaches a critical temperature, a dramatic, collective alignment occurs, manifesting as a spontaneous magnetization. Similarly, in the process of freezing, the molecular disorder of a liquid yields to an ordered crystalline structure as temperature decreases. These transitions are governed by principles like symmetry breaking and the minimization of free energy \cite{goldenfeld2018lectures}. For example, consider the Ising system, phase transition happens at the discontinuity of $M$ the total magnitization of the system as the external field $J$ varies. We can represent $s_1$ as the Ising system and $s_2$ as the system that produced the external field, and $s_1\vee s_2$ as the coupled system. We can represent $\Phi(s_1)$ as the total magnetization of the system $s_1$. Now we have that $\Phi(s_1\vee s_2)$ represents the total magnetization of the coupled system, and $\Phi(s_1)\vee \Phi(s_2)$ represents the average total magnetization of individual systems. The difference between $\Phi(s_1\vee s_2)$ and $\Phi(s_1)\vee \Phi(s_2)$ represents the susceptibility of the system under the change of external fields which captures phase transition.

%Importantly, if this difference is viewed as a kind of information loss, this consideration agrees with certain information-theoretic results for networks. For example, \cite{gershenson2012complexity} has shown through simulations that networks with poor connectivity (shorter attention span) lead to greater information loss. 

~\\

These examples provide a high-level illustration of emergent effects according to our definition in Definition \ref{emergencedef}. In Section 5 and 6 we will demonstrate in our theory and numerical experiments. But first, in the next section, we formulate this approach in a more tractable way in order to produce an effective computational measure.

\section{Emergence as Loss of Exactness}

In this section, we reformulate the intuitive concepts of emergence introduced in the previous sections to derive a quantitative measure of emergence. The technical arguments necessitate an introduction to concepts of category theory and homological algebra. However, we begin by presenting the intuitive ideas.

In Definition \ref{emergencedef} we say that a system is emergent if there exist two components satisfying the condition that $\Phi(s_1\vee s_2) \neq \Phi(s_1)\vee \Phi(s_2)$. However, as discussed in the Introduction, on its own this definition of emergence is not sufficient to evaluate the amount of emergent effects that can present in a system. We need to approximate the difference when $\Phi(s_1\vee s_2)$ and $\Phi(s_1)\vee \Phi(s_2)$ are representing two non-numerical structures, for example, networks. We show in Section 5 that a mapping $\Phi$ that sends a network to its partial structure can be realized as  functors, a fundamental concept in category theory. We can use derived functors to capture the nonlinearity that exists in such functors.

\cite{adam2017systems} adopted the homological algebraic approach to describe emergent effects through interactions among the components, formulated in the language of category theory. The use of category theory in other than pure math is gaining increasing traction and producing some interesting tools and insights \citep[see, e.g.,][]{rosen1958relational, rosen1958representation, bradley2018applied, fong2018seven, northoff2019mathematics}. In \cite{adam2017systems}, the use of category theory is natural since the Abelian category is foundational to the tools and constructions of homological algebra. Importantly, \cite{adam2017systems} proved that emergence can be formulated as "loss of exactness", a notion in homological algebra. This perspective has a number of advantages for quantifying emergence, which we progressively discuss in detail in this section.  

A category is an algebraic structure consisting of objects and morphisms. We consider the collection of subsystems and their constituent components that give rise to emergent effects as the category \textbf{System}, where each object is a subsystem. The morphisms represent a relation between subsystems to be interpreted. For example, $f: V\rightarrow W$ can give a way to embed subsystem $V$ into $W$. When subsystems interact they form a new subsystem (we also call it interconnection). This new resultant subsystem is captured as the mathematical concept of colimit with respect to the interaction diagram, which is treated in detail in \cite{fong2018seven}. Colimit is an object in \textbf{System}, as defined below, that satisfies certain properties. In Section 4 we discuss the category of quiver representations as \textbf{System}  and we will give the colimit and see how it can be regarded as interconnection among objects. 

\begin{defi}
A co-cone of a diagram $F: J\rightarrow C$ in a category $C$ is an object $N$ in $C$ together with a family of morphisms
\begin{align}
\phi_X: F(X) \rightarrow N
\end{align}
for every object $X$ of $J$, such that for every morphism $f: X\rightarrow Y$ in $J$, we have $\phi_Y \circ F(f) = \phi_X$. 
\end{defi}

\begin{defi}
A colimit of a diagram $F:J\rightarrow C$ is a co-cone $(L,\phi_X)$ such that for any other co-cone $(N,\psi)$ of $F$ there exists a unique morphism $u: L\rightarrow N$ such that $u\circ \phi_X = \psi_X$ for all $X$ in $J$.
\end{defi}

There are several candidates of diagrams $F$ to model the interconnected system $s_1\vee s_2$ as the colimit of $F$. One option would be pushouts, namely when $J$ is of the form $\bullet \leftarrow \bullet \rightarrow \bullet$. The interaction of $s_1$ and $s_2$ can be regarded as the pushout diagram:

\[\begin{tikzcd}[row sep=3em]
s_0 \arrow[r] \arrow[d] & s_1 \arrow[d] \\
s_2 \arrow[r] & s_1\vee s_2
\end{tikzcd}\]
where $s_0$ can be interpreted as the part that $s_1$ and $s_2$ share in common after the interaction. In \cite{adam2017systems}, the diagram $F$ is also called the interaction blueprint, and we have the following definition for interaction in the category theory context:

\begin{defi} (Definition 8.3.2 in \cite{adam2017systems})
The system resulting from the interaction along a diagram $F$ is the object $\varinjlim F$ of the colimit of the diagram $F$. 
\end{defi}

The partial observations of the system are represented by a functor $\Phi: \textbf{System}\rightarrow \textbf{Phenome}$, where $\textbf{Phenome}$ is another category whose objects are the observations of objects in $\textbf{System}$. Now we can translate the definition \ref{emergencedef} of emergence into the category theory setting: 

\begin{defi} (Definition 8.3.5 in \cite{adam2017systems})
A functor $\Phi$ sustains emergence effects if and only if the map $\varinjlim \Phi F \rightarrow \Phi \varinjlim F$ is not an isomorphism for some diagram $F$. 
\end{defi}

 In this way, emergence has been redefined as not preserving some colimit. In category theory terms, this is equivalent to saying that the functor $\Phi$ is not a co-continuous functor. In the following, we are going to show that the extent to which colimit is not preserved is reflected as the loss of exactness. 

The framework is based on a special case of this general framework, specifically, we consider the case when both $\textbf{System}$ and $\textbf{Phenome}$ are Abelian categories, a category in which morphisms and objects can be added and in which kernels and cokernels exist and have desirable properties, for a detailed definition, see \citep{mac2013categories}. Basic examples include the category of Abelian groups, and the category of vector spaces over a field $k$. On Abelian categories,  the constructions of homological algebra are available, including exact sequences, to which we can relate an interaction process.

In \cite{freyd1964abelian} Proposition 2.5.3 the following result has been proven:
\begin{prop}
In an Abelian category, the square 
\[\begin{tikzcd}[row sep=3em]
C \arrow[r,"a"] \arrow[d,"b"] & A \arrow[d,"b'"] \\
B \arrow[r,"a'"] & P
\end{tikzcd}\]
is a pushout if and only if the sequence 
\[\begin{tikzcd}[row sep=3em]
C \arrow[r,"{(a,b)}"] &A\oplus B \arrow[r,"{(\begin{smallmatrix} b' \\ -a' \end{smallmatrix})}"] &P \arrow[r] &0
\end{tikzcd}\]
is exact. 
\end{prop}
Here exact is defined as

\begin{defi}
    A sequence \begin{tikzcd}
\cdots \arrow[r] & M_{n+1} \arrow[r,"f_{n+1}"] & M \arrow[r,"f_n"] & M_{n-1} \arrow[r] & \cdots
\end{tikzcd} is exact when it is exact at every term, that is, $\text{\textrm{Im}} f_{n+1} = \ker f_n$ for all $n$.  

%$A$ and $A'$ and $A''$, which means ${\rm Im} a = \ker f$ and ${\rm Im}f = \ker g$ and ${\rm Im} g = \ker b $.
\end{defi}

The notion of exactness is fundamental in homological algebra, and we can extend this to functors:
\begin{defi}
    A functor $\Phi$ is exact when it maps an exact sequence
\begin{equation}
\begin{tikzcd}
0 \arrow[r, "a"] &A \arrow[r,"f"] & A' \arrow[r,"g"] &A'' \arrow[r,"b"] &0
\end{tikzcd} 
\end{equation}
to another exact sequence. In other words, the following sequence is exact:
\begin{equation}
\begin{tikzcd}
0 \arrow{r}{\Phi a} & \Phi(A) \arrow{r}{\Phi f} & \Phi(A') \arrow{r}{\Phi g} & \Phi(A'') \arrow{r}{\Phi b} &0.
\end{tikzcd}  
\end{equation}
Furthermore, $\Phi$ is left exact if 
\begin{equation}
\begin{tikzcd}
0 \arrow{r}{\Phi a} & \Phi(A) \arrow{r}{\Phi f} & \Phi(A') \arrow{r}{\Phi g} & \Phi(A'')
\end{tikzcd}  
\end{equation}
is exact, and $\Phi$ is right exact if 
\begin{equation}
\begin{tikzcd}
 \Phi(A) \arrow{r}{\Phi f} & \Phi(A') \arrow{r}{\Phi g} & \Phi(A'') \arrow{r}{\Phi b} &0
\end{tikzcd}  
\end{equation}
is exact. \label{Prop_exact_emergent}
\end{defi}

In this context, \cite{adam2017systems} proved the following result:
\begin{prop} (Proposition 8.4.3 in \cite{adam2017systems})
Assume the functor $\Phi$ between two Abelian categories is additive and left exact. Then $\Phi$ sustains emergent effects if and only if for some exact sequence $A\rightarrow A'\rightarrow A'' \rightarrow 0$, the sequence $\Phi(A)\rightarrow \Phi(A') \rightarrow \Phi(A'')\rightarrow 0$ is not exact at either $\Phi(A')$ or $\Phi(A'')$.
\end{prop}

What this proposition achieves is encoding emergent effects as a loss of exactness when a sequence in \textbf{System} is mapped into \textbf{Phenome}. We can take advantage of this fact to develop a measure of emergence that is based on a loss of exactness as developed in homological algebra, which describes the mathematical machinery to quantify loss of exactness \cite{gelfand2013methods}.

Homological algebra is often regarded as a mathematical bridge between the worlds of topology and algebra.
It is the study of homological functors and the intricate algebraic structures they entail; its development was closely intertwined with the emergence of category theory. A central concept is that of specific sequences, referred to as chain complexes, that can be studied through both their homology and cohomology. Homological algebra is capable of extracting information contained in these complexes and presenting them in the form of homological invariants of rings, modules, topological spaces, and other 'tangible' mathematical objects which have their own further structure, properties, and interpretations \cite{cartan1999homological, gelfand2013methods}. 

Within this framework of homological algebra,  the loss of exactness can be encoded in the derived functor introduced below. This approach evaluates emergence via the algebraic properties of the functor itself.
Given a function (or more generally a functor) $\Phi$ on a set (category) of networks that represents the processing of information in a system, we can extract information from the "derivative" of the function (derived functor), represented as $L^1\Phi$ or $R^1\Phi$ in the standard literature. We next show that they encode information related to emergence.

Let \textbf{System} and \textbf{Phenome} both be abelian categories and assume $\Phi$ is one-side exact. For example, in this section, we assume $\Phi$ is left exact and not right exact, thus admitting non-trivial right derived functors $R^n\Phi$. Recall that given an object $A$ in \textbf{System}, if we pick an injective resolution, which is by definition an exact sequence:

\begin{equation}
\begin{tikzcd}[row sep=3em]
0 \arrow[r] & A \arrow[r] & I_0 \arrow[r] & I_1 \arrow[r] & I_2 \arrow[r] & \cdots
\end{tikzcd}
\end{equation}
then we obtain a chain complex:
\begin{equation}
\begin{tikzcd}[row sep=3em]
0  \arrow[r] & \Phi( I_0 )\arrow[r] & \Phi(I_1) \arrow[r] & \Phi(I_2) \arrow[r] & \cdots
\end{tikzcd}	\label{complex}
\end{equation}

then $R^i \Phi(A)$ is defined as the $i$th cohomology object of the complex (similarly, $L^i \Phi(A)$ is defined as the $i$th cohomology object of the complex recovered from the projective resolution). Standard textbooks in homological albgera, for example, \cite{rotman2009introduction} have shown that they are indeed functors, and under the assumption that $\Phi$ is left-exact, we have $R^0\Phi(A) = \Phi(A)$. In Section 5 we will give a concrete definition and example that finds $L^1\Phi(A)$ and $R^1\Phi(A)$.  Now note that the chain complex (\ref{complex}) we recovered is not necessarily exact. To measure how much exactness is lost under this functor, we can take advantage of the following standard result in homological algebra \cite[see, e.g.][]{gelfand2013methods}, %which can describe a long exact sequence:
which gives us a long exact sequence,
%{\color{red} Should this be a lemma?}
\begin{prop}
Given an exact sequence 
\begin{equation}
\begin{tikzcd}[row sep=3em]
0  \arrow[r] &A \arrow[r] & A'  \arrow[r] & A'' \arrow[r] & 0
\end{tikzcd}
\end{equation}
in \textbf{System}, the following long exact sequence is exact:
\\
\begin{equation}\begin{tikzcd}[row sep=3em]
0  \arrow[r] &\Phi(A) \arrow[r] & \Phi(A')  \arrow[r] & \Phi(A'') \arrow[dll] \\
&R^1\Phi(A) \arrow[r] & R^1\Phi(A')  \arrow[r] & R^1\Phi(A'') \arrow[dll] \\
&R^2 \Phi(A) \arrow[r] &\cdots.
\end{tikzcd}\end{equation}
\end{prop}

This means, if $R^1\Phi(A) = 0$, the sequence
\begin{equation}
\begin{tikzcd}[row sep=3em]
0  \arrow[r] &\Phi(A) \arrow[r] & \Phi(A')  \arrow[r] & \Phi(A'') \arrow[r] &0
\end{tikzcd} 
\end{equation}
is right exact, which means no emergence is sustained. Thus the first derived functor $R^1\Phi$ encodes the potential of a system to sustain emergent effects. 

This observation provides the foundation for us to mathematically compute the emergence of systems of interest by connecting emergence with a derived functor. The emergence of a specific system $A$ will be determined by the image of this system under the derived functor, represented as $R^1\Phi(A)$. Note that this gives us a global representation of emergence, not specific to the choice of interacting parts $s_1, s_2$ in definition \ref{emergencedef}.  Back to the analogy to derivatives given in this section, we can see that, when $R^1\Phi(A) = 0$ for any $A$ in $\textbf{System}$ category, that means there will be no loss of exactness, hence $\Phi(s_1\vee s_2) = \Phi(s_1)\vee \Phi(s_2)$. In calculus, if $\Phi$ is a function, this equality suggests that $\Phi$ is linear, with zero second-order derivative. Hence we can see that the first-order derived functor $R^1\Phi$ is analogous to the second-order derivative of $\Phi$. 

In this section we assumed $\Phi$ to be left exact, and thus the first right derived functor $R^1\Phi$ encodes  emergence. In the case when $\Phi$ is right exact, it will admit left derived functors $L^n\Phi$ and now $L^1\Phi$ encodes the emergence. We will discuss in detail in Section 5, where we compute $L^1\Phi$ and $R^1\Phi$ explicitly. 

In this section we introduced the homological algebra framework to measure emergence through derived functor, which is an analogue of derivatives in calculus. We note that derived functor can indeed be a generalization of derivatives: derived functor measures the extent to which a functor fails to preserve interactional effect (realized as colimit), analogous to how second-order or higher-order derivatives measure the extent to which a function fails to be linear. The nonlinearty in the functor's behavior, encoded in the derived functor $R^1\Phi(A)$, reflects the nonlinearity in the system's behavior as we go from single components to parts of the system to system itself, hence captures the concept of emergence. 

%Just as the derivative of a function encodes important information about the function's local behavior, derived functors encode important information about the "local" behavior of functors with respect to the operations they interact with.

We need to introduce a particular constraint to this framework before we can use it to study the emergence of any specific system. In the analysis above, we necessarily assumed that the \textbf{System} category is an Abelian category, or is lifted to an Abelian category, because these types of categories have properties that support the formation of exact sequences. From a practical perspective, in order to implement this constraint, in the next two sections we consider the more general category \textbf{System} introduced in Section 2 to be the category of quiver representations. It can be used to model network structures. It is an Abelian category and its homological algebra has been developed \cite[e.g.][]{derksen2017introduction}.

\section{Quiver Representations}

In this section, we introduce quiver representation as an effective way to represent network systems and a bridge that connects the abstract homological algebra framework in the previous section with an applicable measure of emergence. As a pre-step to compute $L^1\Phi(A)$ or  $R^1\Phi(A)$ which encodes the emergence of a system $A$, we realize the category \textbf{System} as the category of quiver representations $\textbf{Rep}(Q)$, where each object in \textbf{System} is a quiver representation of a quiver $Q$. A quiver representation is the mathematical structure obtained by considering network nodes as vector spaces and edges as linear maps. Related work by other authors has explored the connection between quiver representations, neural networks, and data representations \citep[see, e.g.,][]{armenta2021representation, parada2020quiver, ganev2022quiver}.

The use of quiver representations as the \textbf{System} category has the following advantages: first, the category of quiver representations $\textbf{Rep}(Q)$ is itself an Abelian category, the kind of category on which we can carry out homological algebra constructions. Second, plenty of mathematical theories on quiver representations and quiver varieties have been developed that can potentially be applied to the formulation of neural networks \citep[see, e.g.,][]{kirillov2016quiver, jeffreys2022kahler}. Third, the representation of nodes and edges can potentially be used to model information flow on the network and encode the biophysics that govern those networks.

A \textbf{quiver} is a directed graph where loops and multiple arrows between two vertices are allowed, and the representation of a quiver associates nodes with vector spaces and arrows with linear maps (see Figure \ref{quivereg1}, \ref{quivereg2}). Formally, they are defined as follows:

\begin{defi} A quiver $Q$ is a quadruple $Q = (Q_0,Q_1,h,t)$ where $Q_0$ is a finite set of vertices, $Q_1$ is a finite set of arrows, and $h$ and $t$ are functions $Q_1 \rightarrow Q_0$. For an arrow $a\in Q_1$, $h(a)$ and $t(a)$ are called the head and tail of $a$. 
\end{defi}

\begin{defi} We get a \textbf{representation} $V$ of $Q = (Q_0,Q_1,h,t)$ if we attach to every vertex $x\in Q_0$ a finite dimensional $\C$-vector space $V(x)$ and to every arrow $a\in Q_1$ a $\C$-linear map $V(a): V(ta) \rightarrow V(ha)$. 
\end{defi}
 
\begin{figure}[H]
\begin{center}
   \includegraphics[scale=0.5]{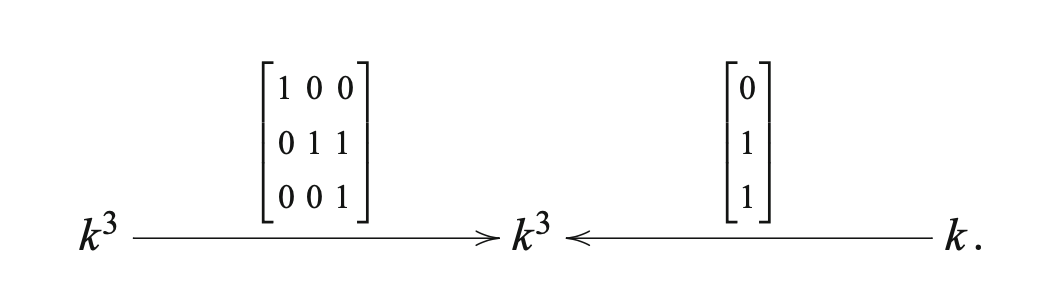} % requires the graphicx package
       
   \caption{A simple example of representation of quiver $\bullet \rightarrow \bullet \leftarrow \bullet$ over a field $k$. Adapted from \cite{schiffler2014quiver}.}
   \label{quivereg1}
   \end{center}

\end{figure}
\begin{figure}[H]
\begin{center}
   \includegraphics[scale=0.25]{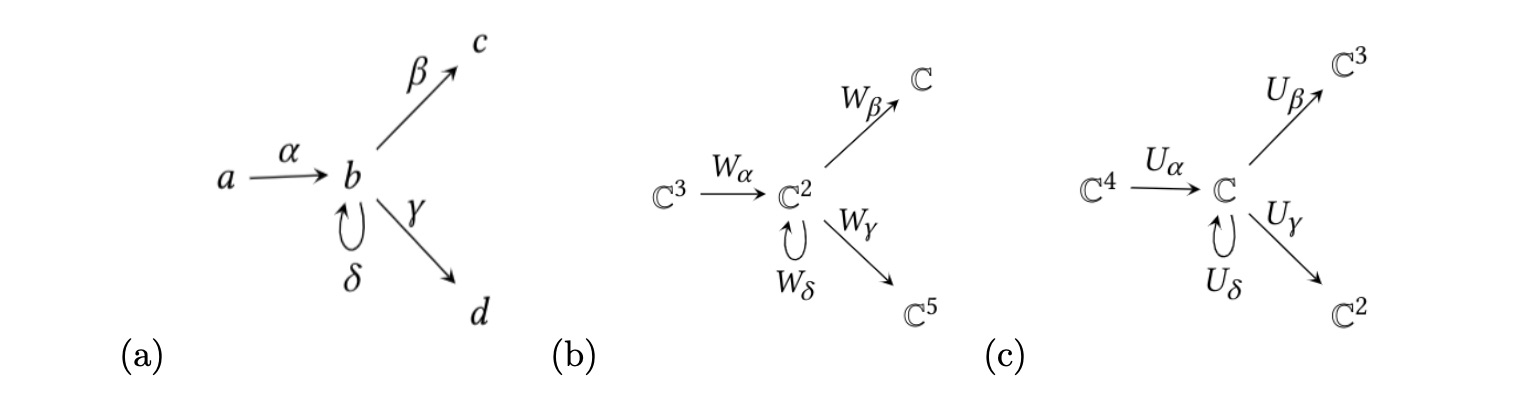} % requires the graphicx package
   \caption{Additional examples of quivers. (a): A quiver $Q$ with vertices $V = \{a,b,c,d\}$ and oriented edges $E = \{\alpha, \beta, \gamma, \delta\}$, (b) and (c): two quiver representations over $Q$. Adapted from \cite{armenta2021representation}. }
   \label{quivereg2} 
   \end{center}

\end{figure}
We now want to formulate the collection of quiver representations into a category in order to apply the homological algebra tools. The category $\textbf{Rep}(Q)$ is defined where the objects are representations of $Q$. For any two representations $V$ and $W$ which are representations of quiver $Q = (Q_0,Q_1,h,t)$, we define a morphism $\phi:V\rightarrow W$ by attaching to every vertex $x\in Q_0$ a $\C$ linear map $\phi(x):V(x)\rightarrow W(x)$ such that for every $a\in Q_1$ the diagram  
\begin{equation}
\begin{tikzcd}%[row sep=8em]
V(ta) \arrow[r, "V(a)"] \arrow[d,"\phi(ta)"']& V(ha)\arrow[d,"\phi(ha)"] \\
W(ta) \arrow[r,"W(a)"'] & W(ha)
\end{tikzcd}
\end{equation}
commutes, i.e., $\phi(ha)V(a) = W(a)\phi(ta)$.

\begin{defi} A \textbf{path}  $p$ in a quiver $Q = (Q_0,Q_1,h,t)$ of length $l\geq 1$ is a sequence $p = a_la_{l-1}\cdots a_1$ of arrows in $Q_1$ such that $ta_{i+1} = ha_i$ for $i = 1,2,\ldots, l-1.$ We define $h(p) = hp = ha_l$ and $t(p) = tp = ta_1.$ For every $x\in Q_0$, we introduce a trivial path $e_x$ of length $0$. We define $he_x = te_x = x$.
\end{defi}

\begin{defi} The \textbf{path algebra} $\C Q$ is a $\C$-algebra with a basis labeled by all paths in $Q$. We denote by $\langle p\rangle$ the element of $\C Q$ corresponding to the path $p$ in $Q$. The multiplication in $\C Q$ is given by 
\end{defi}

\begin{align}
\langle p\rangle \cdot \langle q\rangle =     \left\{ \begin{array}{rcl}
         \langle pq\rangle & \mbox{if}
         & tp = hq, \\ 
         \\
         0  & \mbox{if} &tp \neq hq.
                         \end{array}\right.
\end{align}
Here $pq$ denotes the concatenation of paths, and we use the conventions $pe_x = p$ if $tp = x$, and $e_xp = p$ if $hp = x$.

\begin{prop} The categories $\textbf{Rep}(Q)$ and $\C Q$-mod are equivalent. 
\end{prop}

This result is proven as Proposition 1.5.4 in \cite{derksen2017introduction}. It implies that every quiver representation, including network structures, can be interpreted as a module, an algebraic structure. The category of quiver representations, therefore, behaves similarly as the category of modules. Hence, $\textbf{Rep}(Q)$ is indeed an Abelian category, \cite[see,][]{derksen2017introduction}, so the constructions of homological algebra are available. We can then use tools primarily from linear algebra and graph theory, to compute emergence as the algebraic structure arising from the interactions between quiver representations.  

It would be helpful at this point to give an example of the interaction between quiver representations, which is $V\vee W$, as the colimit (pushout) described in Section 3.  We glue $V$ and $W$ together based on the following diagram:

\begin{equation}
\begin{tikzcd}
V \wedge W \arrow[r, "f"] \arrow[d, "g"'] & V \arrow[d] \\
W \arrow[r]                               & V\vee W    
\end{tikzcd}
\end{equation}
where $V\wedge W$ is the representation that $V$ and $W$ share in common in the inter-connected system $V\vee W$. Consider the quiver $\bullet \rightarrow \bullet$, its representation $V$ as $\C\oplus \C \rightarrow \C \oplus \C$, $W$ as $\C\oplus \C \rightarrow \C \oplus \C$, and $V\wedge W$ as $\C \rightarrow \C$, and $f,g$ are inclusions, where $f: V\wedge W \rightarrow V$ embed to the second component of $V$ while $g: V\wedge W \rightarrow W$ embed to the first component of $W$. Then $V\vee W$ will be $\C\oplus \C \oplus \C \rightarrow \C \oplus \C\oplus \C$, as if we glued $V$ and $W$ together at $V\wedge W$. %We note that the morphisms $f$ and $g$ can be more complicated, as (\ref{morph}) in the next section, which could give more interesting behaviors. 

%{\color{red} Colimit of quiver representations}

\section{Computing the Derived Functors}

In Section 3 we introduced how the cohomology object $R^1\Phi(A)$ is able to encode the emergence of system $A$ in a non-trivial way. In this section we discuss how it is computed, resulting in Proposition \ref{prop55} that leads to a computational measure of emergence.  %Although the language of homological algebra is abstract, we will find eventually the computations are reduced to simple linear algebra. 

The first step is to find the projective/injective resolutions associated with a quiver representation. Let $Q = (Q_0,Q_1,h,t)$ be an acyclic quiver. We investigate the structure of projective and injective resolutions in $\textbf{Rep}(Q)$. For $x\in Q_0$, define $\C Q$-modules  $P_x: = \C Q e_x$ and $I_x  = (e_x \C Q)^*$, the dual space of $e_x \C Q$ defined as ${\rm Hom}(e_x \C Q,\C)$. Here $ \C Q e_x$ gives us the algebra whose basis is given by the paths starting from node $x$.
   
From \cite{derksen2017introduction}, for representation $V$ in $\textbf{Rep}(Q)$  we have the projective resolution 
\begin{equation}
%\[
   \begin{tikzcd}[row sep=3em]
0   &V \arrow{l} & \bigoplus\limits_{x\in Q_0} V(x) \otimes P_x  \arrow{l}[swap]{f^V} & \bigoplus\limits_{a\in Q_1}  V(ta) \otimes P_{ha}  \arrow{l}[swap]{d^V} & 0 \arrow[l]
\end{tikzcd}	\label{reso_1}
%\]
\end{equation}
where 
\begin{align}
  f^V: \bigoplus\limits_{x\in Q_0}  V(x) \otimes P_x \rightarrow V
  \end{align}
is defined by 
\begin{align}
f^V(v\otimes p) = p\cdot v,
\end{align}
and 
\begin{align}
d^V: \bigoplus\limits_{a\in Q_1}  V(ta) \otimes P_{ha} \rightarrow \bigoplus\limits_{x\in Q_0}  V(x) \otimes P_x \label{dv}
\end{align}
is defined by 
\begin{align}
d^V(v\otimes p) =  (a\cdot v)\otimes p- v\otimes pa.
\end{align}

And dually, for representation $W$ in $\textbf{Rep}(Q)$, an injective resolution is% ({\color{red} explain what does injective resolution mean})
\begin{equation}
   \begin{tikzcd}[row sep=3em]
0  \arrow[r] &W \arrow[r,"f_W"] & \bigoplus\limits_{x\in Q_0} I_x \otimes W(x)  \arrow[r,"d_W"] & \bigoplus\limits_{a\in Q_1} I_{ta} \otimes W(ha)  \arrow[r] & 0 
\end{tikzcd}	\label{reso_2}
\end{equation}
   where 
  \begin{align}
  f^*_W: \bigoplus\limits_{x\in Q_0}  W(x)^* \otimes I_x^* \rightarrow W^*
  \end{align}
is defined by 
\begin{align}
f^*_W(p\otimes w^*) = p\cdot w^*,
\end{align}
and 
\begin{align}
d^*_W: \bigoplus\limits_{a\in Q_1}  W(ha)^* \otimes I_{ta}^* \rightarrow \bigoplus\limits_{x\in Q_0}  W(x)^* \otimes I_x^*
\end{align}
is defined by 
\begin{align}
d^*_W(w^*\otimes p) = w^*\otimes ap - (a\cdot w^*)\otimes p.
\end{align}
Here the dual $W^*$ of a quiver representation $W$ in $\textbf{Rep}(Q)$ can be understood as a representation of the opposite quiver $Q^{op}$, obtained by reversing the direction of each arrow in $Q$. 

\begin{prop}
The complexes represented by (\ref{reso_1}) and (\ref{reso_2}) are exact. 
\end{prop}

\begin{proof}
This is proved in \cite{derksen2017introduction} Proposition 2.3.4.

\end{proof}

Following the definition of derived functors in Section 3, the left derived functor $L^n\Phi$ and the right derived functor $R^n\Phi$ can be computed from the following sequences, which are the image of the resolutions under the functor $\Phi$:
\begin{equation}
%\[
   \begin{tikzcd}[row sep=3em]
0   &\Phi(V) \arrow{l} & \bigoplus\limits_{x\in Q_0} \Phi(V(x) \otimes P_x)  \arrow{l}[swap]{\Phi f^V} & \bigoplus\limits_{a\in Q_1}  \Phi(V(ta) \otimes P_{ha})  \arrow{l}[swap]{\Phi d^V} & 0 \arrow[l]
\end{tikzcd}	\label{reso_1_im}
%\]
\end{equation}	
when $\Phi$ is right exact and 
\begin{equation}
 \begin{tikzcd}[row sep=3em]
0  \arrow[r] &\Phi(W) \arrow[r,"\Phi f_W"] & \Phi(\bigoplus\limits_{x\in Q_0} I_x \otimes W(x)  \arrow[r,"\Phi d_W"]) & \Phi(\bigoplus\limits_{a\in Q_1} I_{ta} \otimes W(ha) ) \arrow[r] & 0 
\end{tikzcd}	\label{reso_2_im}
\end{equation}
when $\Phi$ is left exact.  
Following \cite{rotman2009introduction}, this yields
\begin{align}
    L^0\Phi(V) &= H^0 I^\bullet(V) = \Phi(V)\\
    L^1\Phi(V) &= H^1 I^\bullet(V) = \ker \Phi d^V \label{eq1}
\end{align}
when $\Phi$ is right exact and 
\begin{align}
    R^0\Phi(W) &= H^0I^\bullet(W) = \Phi(W)\\
    R^1\Phi(W) &= H^1 I^\bullet(W) = \Phi(\bigoplus\limits_{a\in Q_1} I_{ta} \otimes W(ha) )/ {\rm Im} d_W \label{eq2}
\end{align}
when $\Phi$ is left exact. 

This means to evaluate the emergence of a system, we need to solve equation (\ref{eq1}) or (\ref{eq2}), which will be given in Proposition 5.3.

Now we discuss our choice of functors $\Phi$. Given a class of functors $\{\Phi_\delta\}$, the goal is to find those that could sustain emergent effects, and especially those with stronger potentials for emergent effects. The class of functors $\{\Phi_\delta\}$ of interest will often neglect the partial structure of the system, which resembles information processing systems where some connections are neglected, a ubiquitous class in both natural and engineered systems \citep{mengistu2016evolutionary,udden2020hierarchical}.

To construct edge-deleting functors that sustains emergent effect, we build on the following idea: let $e$ be the edge to be neglected, let $te$ be the source node of $e$ and $he$ be the target node of $e$. Then the functor maps $W(te)$ to its subspace $\ker W(e)$ or the functor maps $W(he)$ to its subspace $W(he)/\text{Im}W(e)$, and then maps $W(e)$ to zero map. Additionally, since mapping the vector space to its subspace may cause other arrow maps to loose domain/image, the functor maps those arrow maps to zero. For $e'$ such that its domain or image overlaps with $W(te)/\ker W(e)$ or with $\text{Im}W(e)$, $W(e')$ is mapped to zero map.

We can see that such functor can indeed sustain emergent effect. Let $V,W$ both be $\C\rightarrow \C$, $f:V\rightarrow W$ takes the following form:
\begin{center}
\begin{equation}
\begin{tikzcd}
\C \arrow[rr, "f_1"] \arrow[d, "V(a)"] &  & \C \arrow[d, "W(a)"] \\
\C \arrow[rr, "0"]                &  & \mathbb{C}         
\end{tikzcd}
\end{equation}
\end{center}

where $a$ is the edge in the quiver, and suppose $f_1: \ker V(a) \oplus ( \C/\ker V(a)) \rightarrow \ker W(a) \oplus (\C/\ker W(a))$, is constructed in the following way:

\begin{center}
\begin{equation}
\begin{tikzcd}
\ker V(a) \arrow[r, "0"] \arrow[d, "\oplus", phantom] &  \ker W(a) \arrow[d, "\oplus", phantom] \\
\mathbb{C} / \ker V(a) \arrow[r, "i_1"]                &   \mathbb{C}/\ker W(a)                  
\end{tikzcd}
\end{equation}\label{morph}
\end{center}

where $i_1$ is the injection map from $\C/\ker V(a)$ to $\C/\ker W(a)$, suppose $\C/ \ker V(a)\subset \C/ \ker W(a)$. Then $\Phi(f_1): \ker V(a) \rightarrow \ker W(a)$ will be zero map while $f_1$ is not a zero map. As a consequence, $\Phi$ will not preserve cokernel, hence not right exact \cite[see, e.g.,][]{rotman2009introduction}, sustaining  emergent effect (see, Proposition 3.8).

\iffalse
 Let $V$ is $\C\oplus \C \rightarrow \C \oplus \C$, and $W$ be $\C\oplus \C \rightarrow \C \oplus \C$, and $V\wedge W$ is $\C \rightarrow \C$, and the restriction of
$f: V\wedge W \rightarrow V$ on the second component will be $f_2: \C \rightarrow \C = \ker a_2 \oplus \C/\ker a_2 \rightarrow \ker f_2 \oplus \C/\ker f_2$, defined in the following way 

\begin{tikzcd}
\ker a_2 \arrow[rr, "0"] \arrow[d, "\oplus", phantom] &  & \ker a_2 \arrow[d, "\oplus", phantom] \\
\mathbb{C} / \ker a_2 \arrow[rru, "i"]                &  & \mathbb{C}/\ker a_2                  
\end{tikzcd}

suppose, for example, $i$ is surjective, then $\text{Im} f_2 = \text{Im} i = \C$ but $\text{Im} \Phi(f_2) = 0$. We now have $V\vee W = \C \oplus (\C \oplus \C / \C ) \oplus \C  \rightarrow \C \oplus \C \oplus \C \oplus \C $, so $\Phi(V\vee W) = \ker a_1 \oplus\ker a_2 \oplus \ker a_2 \oplus \ker a_3 / \text{Im}i  \rightarrow \C \oplus \C \oplus \C \oplus \C $

and $\Phi(V)\vee \Phi(W)$ is $ \ker a_1 \oplus \ker a_2 \oplus \ker a_2 \oplus \ker a_3 \rightarrow \C \oplus \C \oplus \C \oplus \C$, since in the target category, $f_2 = 0$ but $g_2$ is nontrivial. 

Since $\Phi(V)\vee \Phi(W) \neq \Phi(V\vee W)$, we conclude that the system above sustains emergent effects. 

We will show in a following document that $\Phi$ is indeed a functor and it is a non-exact functor. We then compute the first derived functor of $\Phi$.

~\\
\fi

We now show that $\Phi$ is indeed a functor.  We need to show that $\Phi(id_X) = id_{\Phi(X)}$ for every $X$ in the category, and $\Phi(g\circ f) = \Phi(g) \circ \Phi(f)$ for all morphisms $f: X\rightarrow Y$ and $g: Y\rightarrow Z$.  The first condition is satisfied, since $\Phi(id_X)$ consists of identity maps on vector spaces in $\Phi(X)$, and so is $id_{\Phi(X)}$. The second condition is also satisfied. Consider the example below, focusing on a single node, let $f: \C \rightarrow \C$ and $g: \C \rightarrow \C$, defined on $V(ta)$ and $W(ta)$ respectively. As $\Phi$ maps $V(ta)$ to $\ker V(a)$, $\Phi(g\circ f)$ will be the restriction of $g\circ f$ at $\ker V(a)$. For $x\in \ker V(a)$, $\Phi(g\circ f)(x) = g\circ f(x) = g(f(x))$, on the other hand, $\Phi(g)\circ \Phi(f)(x) = \Phi(g)(\Phi(f)(x)) = \Phi(g)(f(x))$. Since $\Phi(g)$ is the restriction of $g$ at $\ker W(a)$, we now need to show that $f(x)\in \ker W(a)$. The morphisms by definition has to satisfy the commutative diagram with the arrow map, so for $x\in \ker V(a)$, we must have $W(a)\circ f(x) = h \circ V(a) =  0$, which means $f(x)\in \ker W(a)$. Hence $\Phi(g)(f(x)) = g(f(x))$, and we conclude $\Phi(g\circ f) = \Phi(g)\circ \Phi(f)$.
\begin{center}
\begin{equation}
\begin{tikzcd}
\C \arrow[rr, "f"] \arrow[d, "V(a)"] &  & \C \arrow[d, "W(a)"] \arrow[rr, "g"] & &\C  \arrow[d]\\
\C \arrow[rr, "h"]                &  & \mathbb{C} \arrow[rr] && \mathbb{C}       
\end{tikzcd}
\end{equation}
\end{center}

% \begin{defi}
% A set of removable edges $E$ is a set of edges that are disjoint from each other, so the
% \end{defi}

Formally, also consider the multiple edges case, the functor we constructed can be defined in the following way: given a set $E$ of edges to be neglected, we consider the functors:  

\begin{itemize}
\item $\Phi_l: \textbf{Rep}(Q) \rightarrow \textbf{Rep}(Q)$, which maps $V(x)$ to $V(x)/ \bigcup_{he = x, e\in E} \text{Im} V(e)$, where $\bigcup_{he = x, e\in E} \text{Im} V(e)$ is the union of images of linear maps associated to all edges in $E$ with head $x$ (coming into $x$). $\Phi_l$ maps $V(e), e\in E$ to zero maps. For any edge $e'$ such that its domain or image overlaps with $\bigcup_{he = x, e\in E} \text{Im} V(e)$, $V(e')$ is mapped to zero map. For morphisms, $\Phi_l$ maps the morphisms between quiver representations $f:V\rightarrow W$ to their restrictions on $\Phi_l(V)$, $f|_{\Phi_l(V)}: \Phi_l(V) \rightarrow \Phi_l(W)$.

\item $\Phi_r: \textbf{Rep}(Q) \rightarrow \textbf{Rep}(Q)$, which maps $V(x)$ to $ \bigcap_{te = x, e\in E} \ker V(e)$, the intersection of images of linear maps associated to all edges in $E$ with tail $x$ (coming from $x$). $\Phi_l$ maps $V(e), e\in E$ to zero maps. For any edge $e'$ such that its domain or image overlaps with $V(x)/\bigcap_{te = x, e\in E} \ker V(e)$, $V(e')$ is mapped to zero map. For morphisms, $\Phi_r$ maps the morphisms between quiver representations $f:V\rightarrow W$ to their restrictions on $\Phi_r(V)$, $f|_{\Phi_l(V)}: \Phi_r(V) \rightarrow \Phi_r(W)$. 
\end{itemize}

We can see that $\Phi_l$ and $\Phi_r$ equivalently delete edges in $E$ and potentially some other affected edges from a quiver representation. We use $E'$ to represent all edges in $V$ whose arrow maps are mapped to zero by $\Phi_l$ or $\Phi_r$. In practice, the quiver representation can be constructed to eliminate overlaping domains/images, such that $E = E'$.  In the following we work on the category of quiver representation $\textbf{Rep}(Q)$ where we assume $Q$ is a finite quiver with non-trivial edge set, and the functor $\Phi_l$ and $\Phi_r$ exists.

Note that from a modeling perspective, these functors are also related to the mapping from one scale to a higher scale in a real-world system. When doing a higher-level partial observation, some information about the lower-level system will be missing, for example, some edges and nodes in the network. Such functor might be further modified into one that represents a coarse-graining scheme or a convolutional architecture.

\begin{prop}
$\Phi_l$ is left exact and $\Phi_r$ is right exact.
\end{prop}

\begin{proof}
It's easy to check that $\Phi_l$ and $\Phi_r$ are additive functors. Proposition 5.25 in  \cite{rotman2009introduction} shows a functor $F: \textbf{Mod}\rightarrow \textbf{Ab}$ is left exact when it preserves kernel and right exact when it preserves cokernel. We can see that for a morphism $f: V(x)\rightarrow W(x)$, we must have $\text{Im}f \subset \ker W(a)$, where $a$ is an edge being removed by the functor, with $ta = x$. This means when restricting $f$ to $\ker V(a) \rightarrow \ker W(a)$, injectiveness will be preserved while surjectiveness might not be preserved; when restricting $f$ to $V(x)/\ker V(a) \rightarrow W(x)/\ker W(a)$, surjectiveness will be preserved while injectiveness might not be preserved. Hence $\Phi_r$ preserves cokernel and is right exact, $\Phi_l$ preserves kernel and is left exact. 

\end{proof}

Now we give the following proposition that quantifies emergence as such algebraic structure:

\begin{prop}\label{prop55}
Given a set $E$ of edges to be neglected, the first left derived functor of $\Phi_r$ is 
\begin{align}
L^1 \Phi_r(V) = \bigoplus_{a\in E'_r}\Phi_r(V(ta)\otimes P_{ha})
\end{align}
where $E'_r$ is the set of edges whose arrow maps in $V$ are mapped to zero by $\Phi_r$, $P_{ha}$ is the path algebra spanned by all paths from the head of the edge $a$.

The higher left derived functors of $\Phi_r$ are trivial. 

The first right derived functor of $\Phi_l$ is
\begin{align}
R^1 \Phi_l(W) = \bigoplus_{a\in E'_l}\Phi_l(W(ha)\otimes I_{ta})
\end{align}
where $E'_l$ is the set of edges whose arrow maps in $W$ are mapped to zero by $\Phi_l$, $I_{ha}$ is the path algebra spanned by all paths to the tail of edge $a$.

The higher right derived functors of $\Phi_l$ are trivial. 
\end{prop}

\begin{proof}
First, since the injective and projective resolution in the category of quiver representations has length $2$, we can see that all higher derived functors will vanish. 

We already proved in the previous proposition that $\Phi_r$ is right exact, so we compute the zeroth left derived functor is itself (see, for example, \cite{rotman2009introduction}). Now we compute the first left derived functor $L^1\Phi_r$. From (\ref{eq1}), $L^1\Phi_r = \ker \Phi_r d^V$, where $d^V$ is defined in (\ref{dv}).  Now if an edge $a$ is deleted by the functor $\Phi_r$ then for any $v\in V(ta)$ and $p\in P_{ha}$, we have$(a\cdot v)\otimes p = v\otimes ap = 0$, hence $\Phi(V(ta)\otimes P_{ha}) \subseteq \ker \Phi_r d^V$. If $a$ is preserved under $\Phi_r$, then $\Phi_r d^V$ will act the same as $d^V$ on $\Phi(V(ta)\otimes P_{ha})$, and $d^V$ is injective due to the exactness of resolution,   $\Phi_r(V(ta)\otimes P_{ha})$ will be non-zero thus not contribute to $\ker \Phi_r d^V$.

Now for the left exact functor $\Phi_l$, by definition its first right derived functor $R^1\Phi_l$ is computed as $\Phi_l(\bigotimes_{a\in Q_1}I_{ta}\otimes W(ha))/\text{Im}d_W$ from the injective resolution, and since the injective resolution is obtained from the projective resolution by taking the dual, we have $R^1\Phi_l(W) = L^1\Phi_r(W^*)$ where $W^*$ the dual of $W$ is the opposite quiver representation of $W$ obtained by reversing the direction of all edges. Under this choice, $W^*(ha) = W(ta)$, and $P^*(ta) = I_{ta}$.

\end{proof}

This proposition shows under the choice of the functor that deletes some edges in the network, emergence as $L^1\Phi(A)$ or $R^1\Phi(A)$  encodes information about the nodes and the paths being affected by deleting these edges. In the next section, we show how we can build a computational measure of emergence from such structures.

%the edges thus being neglected. It consists of the direct product of path algebras, spanned by some paths that go through the edges being neglected. So we can interpret the algebraic structure we computed as the part of the quiver representation that is impacted by the functor or neglected by the partial observations, which is responsible for the emergent structure/ pattern. 

%It is also worth pointing out that, we can generalize this proposition to the situation where the functor is reducing the dimensionality of $W_{in}$ or $W_{out}$ instead of mapping it to zero. This requires us to introduce more grade/filter structure to the vector spaces, and this will be discussed in our following work that studies how local structures impact the global emergence. 

\section{A Numerical Measure of Emergence}

In proposition \ref{prop55},  the emergence of system $A$ is encoded as the cohomology  $L^1\Phi(A)$ or $R^1\Phi(A)$. It will be torsion-free\footnote{based on our choice of base field, which is usually $\C$ or $\R$.}, hence the information is encoded in its dimension, which can be regarded as a numerical measure to evaluate the emergence of a system. 
\begin{align}
\dim L^1\Phi_r(V) &= \dim \bigoplus_{e\in E'_r}\Phi_r(V(te)\otimes P_{he})\notag\\
%&= \sum_{a\in E} 
&= \sum_{e\in E'_r} \dim \Phi_r(V(te))\times \dim \Phi_r(P_{he}) \label{measure_1}
\end{align}
Similarly, 
\begin{align}
\dim R^1\Phi_l(W)
&= \dim \bigoplus_{e\in E'_l}\Phi_l(W(he) \otimes I_{te})
\notag\\
&=  \sum_{e\in E'_l} \dim \Phi_l(W(he))\times \dim \Phi_l(I_{te})\label{measure_2}
\end{align}

Here $\dim\Phi_r(V(te))$ and $\dim\Phi_l(W(he))$ is the dimension of the image of the vector space $V(te)$ and $W(he)$ under the functor, and  $\dim \Phi_r(P_{he})$ and $\dim \Phi_l(I_{te})$ is the dimension of the image of the path algebra $P_{he}$ and $I_{te}$ under the functor. As discussed in the previous section, the quiver presentation can be constructed to eliminate overlaping domains or images, such that $E'_r = E'_l = E$,  so hereafter we use $E$ to represent the set of edges being effectively neglected by $\Phi_r$ or $\Phi_l$.

In order to understand this computational measure, it is important to understand the path algebra $P_x$ here. $P_x$ can be regarded as a quiver representation where the vector space associated with each node has the same dimension as the number of path from $x$ to this node \cite{derksen2017introduction}. So we can understand that the dimension  $\dim P_x$ is the number of paths originating from $x$. For the dimension $\dim \Phi(P_x)$, since $\Phi$ does not destroy paths (it just maps the linear map associated to edges to zero, which is distinct from removing the edge from the quiver), $\Phi(P_x)$ should have the same basis as $P_x$. Still, in practical considerations, it should be helpful to consider the path with nonzero linear maps associated to it, in this case, $\dim \Phi(P_x)$ will be interpreted as the paths from $x$ that doesn't contain a deleted edge in $E$. 

%can represent the number of paths being affected  is the number of paths originating from $x$ that are preserved by 
%the functor. The functor kills a path when at least one edge in this path is also in $E$, the set of edges deleted by the functor.  From this we can also see that $L^1\Phi(A)$ or $R^1\Phi(A)$ can be regarded as a quiver representation as well. 

%Note that here we are summing over the dimension of vector spaces associated with all the nodes in the quiver representations corresponding to $L^1\Phi_r(A)$ or $R^1\Phi_l(A)$ in order to have a global measure of emergence. By evaluating first the dimension of each vector space, we can find the contribution of each node to the global emergence, which is an important feature of our framework of emergence. It addressed the issue we discussed in the introduction part of this paper. 

Now we discuss how this measure can be applied to networks. A direct graph can be regarded as a quiver representation where the vector space associated with each node is 1-dimensional. This kind of quiver representation is considered in, for example, \citep[][]{armenta2021representation, ganev2022quiver}. Now based on equations (\ref{measure_1}), we can neglect the $\dim \Phi_r(V(te))$ term, and the measure of emergence will base only on the path algebra term, $\dim \Phi_l(P_{he})$, which is the number of paths from $he$ that does not contain edges in $E$. Since the vector space associated to the vector spaces are 1-dimensional, if a node $x$ is the head of an edge in $E$, then $V(x)$ will be zero in this case. Let $H$ be the node sets of $\Phi(G)$, then $H$ contains the nodes that are not the head of any $e\in E$. The paths in $\Phi_l(P_{he})$ will be all paths in $\Phi(G)$ from $he$ to nodes in $H$. We now have an edge-based measure of emergence:

%Consider the $H$ as the subgraph of $G$ where edges in $E$ are deleted. Then we come to the following measure: 

\begin{align}
\text{Emergence}(G,H) %&= \sum_{e\in E}  \dim \Phi_r(P_{he})\\
=\sum_{e \in E} \#\text{paths from $he$ to $H$}.  \label{networks_edge}
\end{align}

%where 

To come up with node-based measure of emergence, we consider the functor $\Phi$ which deletes all edges coming into a set of nodes, hereby represented as $G\setminus H$. This functor effectively deletes nodes from the network. In this case,  we have a measure of emergence when we pass from a graph $G$ to its subgraph $H$: 
\begin{align}
\text{Emergence}(G,H) &= \sum_{x\in G \setminus H} \#\text{paths from $N_H(x)$ to $H$},  \label{networks}
\end{align}
where $N_H(x)$ is the set of downstream neighbors of $x$ in $H$, and the paths only include nodes in $H$. Note that this result assumes the functor $\Phi$ as deleting nodes in $G \setminus H$. The measure of emergence here could have different forms under other choices of $\Phi$, which will be explored in future work. 

This result quantifies emergence, a kind of structural nonlinearity as given in Definition 2.1, as counting the number of paths generated from the nodes not preserved by the partial observation. It provides us with a numerical measure of emergence that is based on the local network structure, which is the path originating from each node in $G\setminus H$. When we increase the size of network $G$, the number of paths often grows exponentially, which resembles the scaling laws that are present in a lot of systems where emergent behaviors are observed.

Applying Proposition \ref{prop55}, when the functor $\Phi$ neglects more edges, it could imply stronger potential for emergence. This 
conceptually agrees with other existing measures. For example in \cite{gershenson2012complexity} emergence is defined as the information loss
\begin{align}
E = \frac{I_{\text{out}}}{I_{\text{in}}} \label{info_loss}
\end{align}
where $I_{\text{in}}$ is the "input information" and $I_{\text{out}}$ is the "output information", which can be seen as $I_{\text{in}}$ transformed by a computational process. \cite{gershenson2012complexity} simulated random Boolean networks with different connectivity $K$ and found that the loss of emergence increases with $K$. Larger $K$ corresponds to fewer neglected edges, which is equivalent to smaller $L^1\Phi(A)$. This means that $\Phi(s_1\vee s_2)$ is closer to $\Phi(s_1) \vee \Phi(s_2)$, which implies less information loss and greater $E$. 

Importantly, our measure associates the global property of emergence with specific nodes and edges. Compared to information-theoretic approaches, our measure provides mechanistic insights into the subsystems responsible for observed emergent phenomena. We plan to study this in detail in future work. 

\iffalse
In summary, the approach we developed for evaluating the emergent effect of a network system $A$ can be sketched as follows:

\begin{enumerate}
    \item Encode the system $A$ as a quiver representation, i.e., a network where each node is associated with a vector space, and each edge a linear map.
    \item Realize the partial observation as a functor $\Phi$ which neglects the input/output edges of chosen nodes, for example, those that are inactive in a computational process.
    \item Use Proposition \ref{prop55} or Equation (\ref{networks}) to obtain an algebraic/numerical evaluation of emergence. 
\end{enumerate}
\fi

We have developed a measure for emergence, and now we demonstrate numerically that the measure offers computational insights for network systems, modeled as quiver representations. We consider random Boolean network \citep[see,][]{kauffman1969metabolic,gershenson2004introduction},  a type of mathematical model used to represent genetic regulatory networks and study their dynamics. These networks consist of nodes (or vertices) and directed edges (or links). Each node in the network can be in one of two states, typically represented as 0 or 1 (or "off" and "on").  Random Boolean network can be characterized by the connectivity parameter $K$, where each node in the network receives input from $K$ nodes. In the setup, each node is connected to (receives input from) randomly chosen $K$ nodes. The initial state and the update rule, which decides the state of each node based on a given input, are generated randomly.

\begin{figure}[H]
\begin{center}
   \includegraphics[scale=0.48]{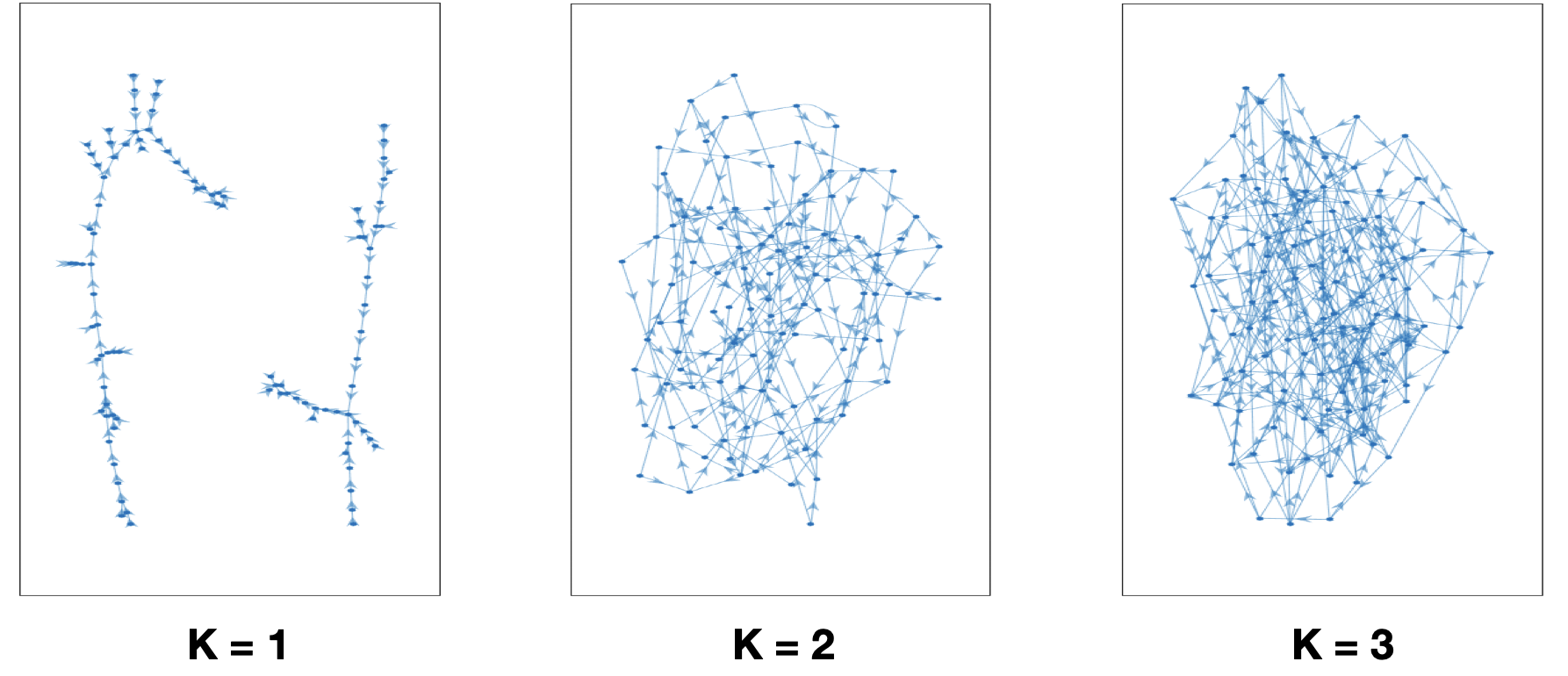} % requires the graphicx package
   \caption{Examples of random Boolean network for connectivity $K = 1,2,3$. }
   \end{center}
\end{figure}

Random Boolean network can be regarded as a generalization of cellular automata \cite{gershenson2004introduction}, and has been used as models to study emergent behavior, especially in the context of biological and complex systems\citep[see,][]{kauffman1993origins, gershenson2003classification}. In \cite{gershenson2012complexity}, random Boolean networks of different connectivity has been simulated, and the information loss as in (\ref{info_loss})
\iffalse
\begin{align}
E = \frac{I_{out}}{I_{in}} \label{info_loss}
\end{align}
\fi
has been proposed as a measure of emergence, where $I_{out}$ is the Shannon information measured at the end state of the network and $I_{in}$ is the Shannon information measured at the initial state of the network. To compare our measure of emergence with this measure, we run our simulation on networks with a range of connectivities $K$ where the edges and activation functions are randomly generated. We measure emergence of such random Boolean networks by regarding the functor $\Phi_l$ as the one that sends a fully connected network $G$ to subnetwork $H$ which remains active during the network dynamics. A node is considered inactive here when it comes into state $1$ in less than $5\%$ of the time after the network dynamics reaches steady state. We consider the quiver representation for each network that has dimension $1$ for the input and output space of each node, and as discussed above, and emergence reflects the number of paths being "impacted" by the functor. We compared our node based measure of emergence with the information-theoretic measure of emergence as information loss in Figure \ref{emergence_infoloss_plot}.
% \begin{figure}[H]
% \begin{center}
%    \includegraphics[scale=0.6]{fig/Emergence_info_loss_nov27_end.eps} % requires the graphicx package
%    \caption{Comparison of emergence and information loss at different network connectivity $K$. }
%    \label{emergence_infoloss_plot}
%    \end{center}
% \end{figure}

% \begin{figure}[H]
% \begin{center}
%    \includegraphics[scale=0.6]{fig/Emergence_vs_node_loss_nov27_23_end.eps} % requires the graphicx package
%    \caption{Emergence and the percentage of inactive nodes at different network connectivity $K$. }
%    \label{emergence_nodeloss_plot}
%    \end{center}
% \end{figure}

\begin{figure}[H]
    \centering
    \begin{subfigure}[b]{0.49\textwidth}
        \centering
        \includegraphics[width=\textwidth]{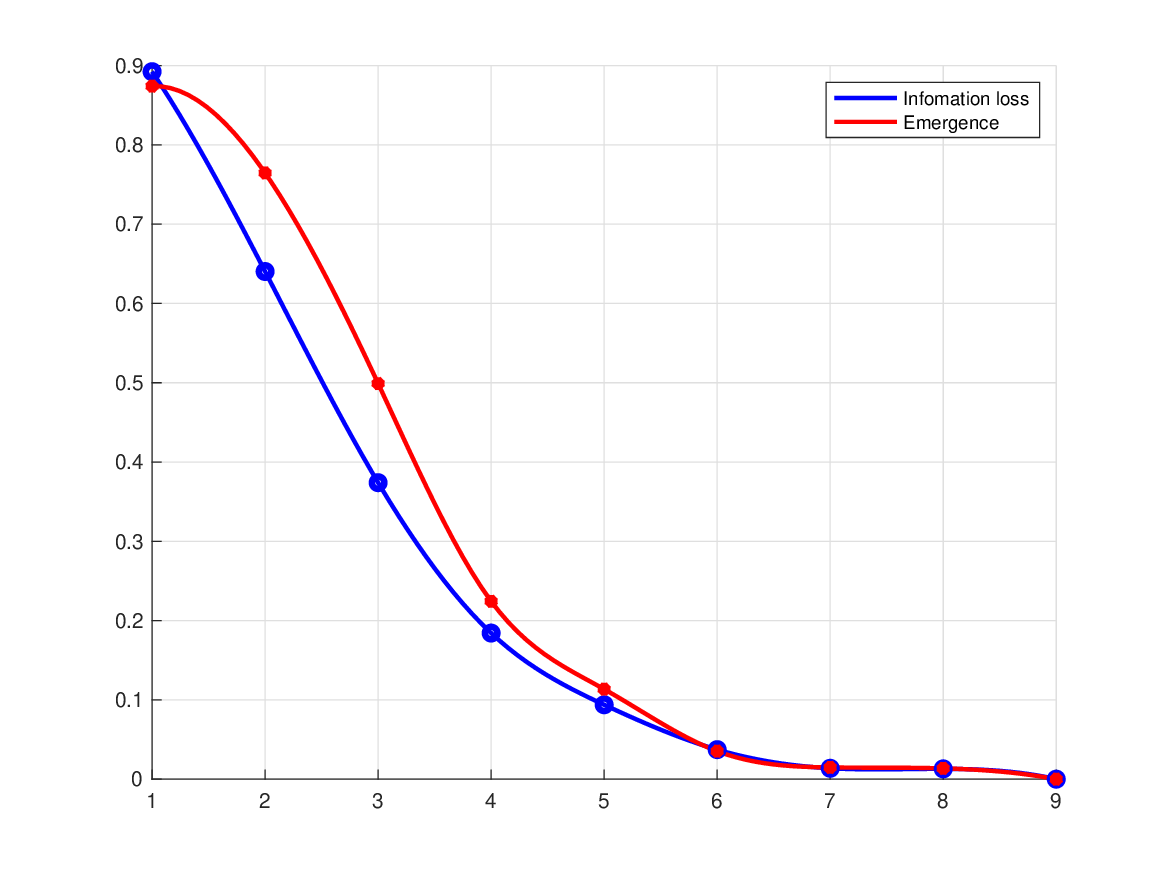}
        \caption{Comparison of emergence and information loss at different network connectivity $K$.}
        %\label{fig:1}
        \label{emergence_infoloss_plot}
    \end{subfigure}
    \hfill
    \begin{subfigure}[b]{0.49\textwidth}
        \centering
        \includegraphics[width=\textwidth]{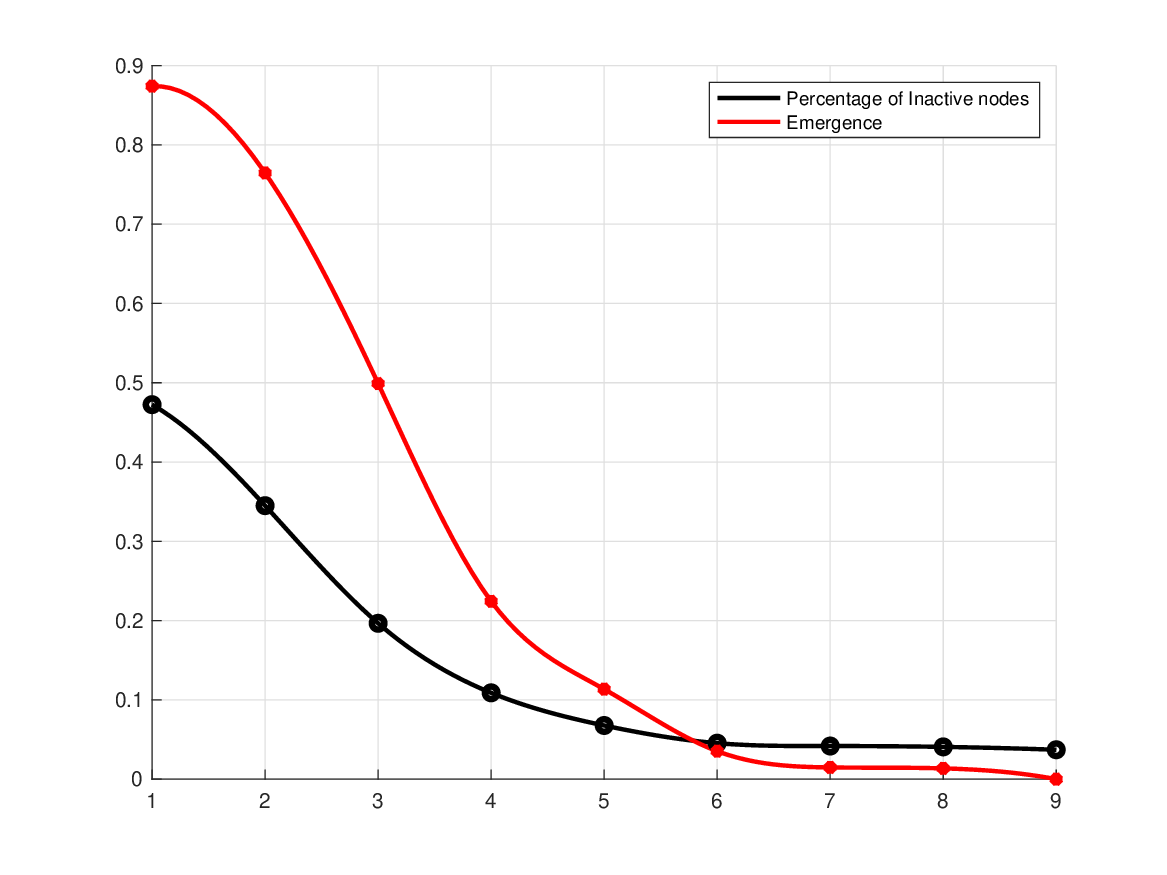}
        \caption{Emergence and the percentage of inactive nodes at different network connectivity $K$. }
        %\label{fig:2}
        \label{emergence_nodeloss_plot}
    \end{subfigure}
    \caption{Numerical results showing the correlation of emergence and information loss}
    \label{fig:main}
\end{figure}

In our computation, we averaged over the random network ensemble of size $1000$. The network has $10$ nodes in total, when $K=9$ the network is fully connected. The initial state of the network is randomized and the probability for each node to be active is $1/2$. We simulate the network for $200$ time steps and it is observed that the network dynamics reach steady state in fewer than $100$ time steps. The emergence is normalized by the maximum and minimum value among all possible value of $K$. The normalized information loss is computed here as 
\begin{align}
E' = 1- \frac{I_{K}}{I}
\end{align}
where $I_K$ is the Shannon information of the network with connectivity $K$ when it reaches steady state and $I$ is the Shannon information of the fully connected network when it reaches steady state. If we regard the computational process as follows: the input is the entropy of the fully connected network and the output is the entropy of the $K$ connected network, then $I_{K}/I$ is interpreted as $E$ in Equation (\ref{info_loss}) in this way. 

From the simulation, we can see that our measure of emergence correlates with information loss, as a way to measure emergence on the random Boolean networks of different connectivity.  We think there is a connection between our measure of emergence and information loss, although it is not proved at the current stage. As an interpretation, the functor $\Phi$ we choose sends the fully connected network $G$ to the subnetwork $H$ of active nodes. Since the active nodes are given by the steady state of network dynamics, the functor ties the steady state of some connectivity $K$ network to the steady state of the fully connected network. Emergence in this context thus captures the difference between the steady states, and more generally the impact connectivity $K$ has on network dynamics.

We plot emergence together with the percentage of inactive nodes in Figure \ref{emergence_nodeloss_plot}. We can observe that emergence largely correlates with the percentage of inactive nodes in a random Boolean network. Although emergence is dependent on the specific positions of the inactive nodes, we can still see a correlation between emergence and the average number of inactive nodes. This agrees with our intuition in Section 2 and equation (\ref{networks}). 

% We can also see that, our measure of emergence suggests that the network has strongest emergence when the connectivity is neither too dense or too sparse (in the example of random Boolean networks, when $K=2$), like a small-world type network \cite{latora2001efficient}. 

\section{Discussion and Conclusion}

%{\color{red} emergence as structural nonlinearity}

Emergence, generally conceived as the novel dynamics, patterns or behaviors of a system that are not present in the constituent parts of the system. We adopt a universal and abstract definition of emergence which captures the structural nonlinearity in system's effect or observation as we go from components to the whole system. We have established a framework that quantifies the emergence of systems based on their structural information, modeled by quiver representations, and we build a computational measure that quantifies the emergence of systems realized as the mapping between networks where only partial nodes and edges are being preserved, which can be used to model a large variety of systems characterized by the information flow on the network. Other kinds of information flow, for example, those involving coarse-graining,  can also be studied by constructing the functors that capture the information flow and compute their derived functors. The emergence is formulated in homological algebra as the "loss of exactness", the loss of the algebraic property when a sequence of subsystems modeling interactions is mapped to their partial observations. We showed that our measure correlates with the information-theoretic measure of emergence on random Boolean networks. One important feature of our measure is, as discussed in Section 6, it offers a way to study the contribution of each node to emergence, which is a global property. This could allow us to track the contribution of each node to emergence in a dynamic process, which can help detect the mechanisms of emergence.

Although this framework is algebraic in nature, we have shown that it can lead to computational measures that are applicable to real-world systems. The key is to choose the appropriate categories and functors to work on that have both relatively good algebraic properties and strong descriptive power to real-world systems. For example, if researchers aim at studying the emergent abilities of machine learning systems, the functor could be the one that ties the systems (for example, parts of the neural network architecture) to its performance, for example, on some data set and downstream tasks. The fundamental constraints within the framework is that $\textbf{System}$ and $\textbf{Phenome}$ have to be Abelian categories (or the extension of it that allows the discussion of homological algebra) and the functor $\Phi$ has to be exact at one side (either left exact or right exact). The first constraint is tackled by considering the category of quiver representations, which is Abelian, with nice homological algebra properties and can encode network systems. The second constraint is tackled by adopting the functor that truncates only part of the vector spaces corresponding to the kernel or cokernel of the linear maps, which can often be one-side exact. By these constructions we allow our framework to be applicable to a large class of real-world systems. 

We also leave a lot of research directions to be investigated later. To better capture emergence, it is possible to impose additional constraints on derived functor $R^1\Phi$, or study the difference in the derived functor properties in different emergent phenomenon. We showed that in the category of quiver presentation, we only have non-trivial first-order derived functors. It will be interesting to study the behavior of the higher-order derived functors and interpret them. When the $\textbf{System}$ category is no longer an Abelian category, we can either lift it to Abelian category, as discussed in \cite{adam2017systems}, or develop homological algebra in a non-Abelian case. It is in general an intriguing task to find new categories to work on to yield insights in network science or other disciplines in science and engineering. Also, the mathematical theory of quiver representations is itself a massive topic, with growing applications to physical systems and network theory, which will be an interesting exploration as well to combine with the homological algebra approach.

From a modeling and computational perspective, our work provides an approach to quantify the emergence of systems as networks, or more generally, quiver representations. It can capture not only the connectivity patterns of the components, but also the structures and dynamics at each node and edge. Within the framework of quiver representations, the mappings between vector spaces need to be linear maps, which renders the framework difficult to capture the non-linearity, like the activation function in artificial and biological neural networks. But still, quiver representation can be applied to model data and its network structures, so our work can be applied to the phase space of the dynamics, like the spike train data and cellular automata. Also, there has been theoretical development to model neural networks based on quiver representations \citep[see][]{jeffreys2022kahler,armenta2021representation}. Additionally, the computation of emergence based on networks or quiver representations would often involve counting the number of certain paths in the network, which would require searching algorithms that are time-consuming if the network is large. We call for new approximation schemes and algorithms that evaluate emergence in more efficient ways.

To our knowledge, this is the first work that gives a computational measure of emergent effects in a system and its constituent parts, thus opening up exciting new research avenues that lead to an understanding of the emergence mechanisms of systems of interest. The computational approach to emergence should appeal to people in a wide range of fields where there is a notion for emergence. It should be mentioned here that the term "generativity" can almost be used interchangeably with emergence. We choose to use emergence since it is a term widely used across fields while "generativity" is mostly discussed in artificial intelligence. 

The work in this paper could potentially apply to, for example, the optimization of network architecture that encourages or discourages emergence so as to achieve better performance in machine learning tasks, and the study of the cause of neurological disorders from the activity and morphology of neurons. In future work, we will apply this framework to the study of the emergent effects of biological and machine learning systems.  For networks in these fields, the nodes and edges often have intricate internal structures that traditional models are not able to capture. We will model these kinds of networks as quiver representations or other mathematical structures. There are topics in these fields that are related to emergence, like phase transition, symmetry breaking and generalizability of artificial neural networks, and we aim to establish the connections between emergence and these concepts as well.

\subsection*{Acknowledgments}
%The people you want to acknowledge. For this document, we appreciate Jörg Lücke, author of an accepted paper who generously allowed us to use his template.

The authors would like to thank Dr. Elie Adam for the extensive conversations on the theoretical foundation of this work. This work was supported by unrestricted research funds to the Center for Engineered Natural Intelligence (CENI) at the University of California San Diego.

\iffalse
\section*{Appendix}
You should put the details that are not required in the main body into this Appendix.
\fi

\bibliographystyle{APA}

\end{document}